\documentclass[aps,superscriptaddress, twocolumn]{revtex4}
\usepackage{amsmath}
\usepackage{shortcuts}
\usepackage{complexity}
\usepackage{color}
\usepackage{datetime}
\usepackage{subfigure}
\usepackage{hyperref}
\usepackage[utf8]{inputenc}
\usepackage{tikz,tikz-3dplot}
\usepackage{comment}
\usetikzlibrary{decorations.markings}
\usepackage{leftidx}
\usepackage{algorithm}
\usepackage[noend]{algpseudocode}
\usepackage[normalem]{ulem}

\input{Qcircuit}

\def\weak{\ComplexityFont{WEAK}}
\def\strong{\ComplexityFont{STR}}
\def\etal{\textit{et al}.\ }

\begin{document}
\title{Computing quopit Clifford circuit amplitudes by the sum-over-paths technique}
\author{Dax Enshan Koh}
\email{daxkoh@mit.edu}
\affiliation{Department of Mathematics, Massachusetts Institute of Technology, Cambridge, Massachusetts 02139, USA}
\author{Mark D. Penney}
\email{mpenney@mpim-bonn.mpg.de}
\affiliation{Max Planck Institute for Mathematics, Vivatsgasse 7, 53111 Bonn, Germany}
\author{Robert W. Spekkens}
\email{rspekkens@perimeterinstitute.ca}
\affiliation{Perimeter Institute for Theoretical Physics, 31 Caroline Street North, Waterloo, Ontario, Canada N2L 2Y5}

\begin{abstract}
By the Gottesman-Knill Theorem, the outcome probabilities of Clifford circuits can be computed efficiently. We present an alternative proof of this result for quopit Clifford circuits (i.e., Clifford circuits on collections of $p$-level systems, where $p$ is an odd prime) using Feynman's sum-over-paths technique, which allows the amplitudes of arbitrary quantum circuits to be expressed in terms of a weighted sum over computational paths. For a general quantum circuit, the sum over paths contains an exponential number of terms, and no efficient classical algorithm is known that can compute the sum. For quopit Clifford circuits, however, we show that the sum over paths takes a special form: it can be expressed as a product of Weil sums with quadratic polynomials, which can be computed efficiently. This provides a method for computing the outcome probabilities and amplitudes of such circuits efficiently, and is an application of the circuit-polynomial correspondence which relates quantum circuits to low-degree polynomials.

\end{abstract}
\maketitle

\section{Introduction}

Computing the outcome probabilities of a quantum circuit is in general a hard problem. In complexity-theoretic terms, it belongs to a class of problems known as $\#\P$-hard, which are widely conjectured to not be efficiently solvable by a classical computer (or even a quantum computer) \cite{Arora}. Nevertheless, there are interesting subclasses of quantum circuits for which we do know efficient classical algorithms to compute the outcome probabilities.  An example is the class of (nonadaptive) Clifford circuits, which has been studied extensively in quantum information theory, for example, in quantum error correction \cite{GottesmanThesis} and in measurement-based quantum computation \cite{Raussendorf, RaussendorfBriegel}. These circuits are rich enough to exhibit many of the `nonclassical' features of quantum mechanics like entanglement and quantum teleportation, but yet are not rich enough to preclude efficient simulation by a classical computer \cite{AaronsonGottesman}. The latter fact is the content of the Gottesman-Knill Theorem \cite{gottesman1998heisenberg}, and one of its implications is the existence of an efficient classical algorithm to compute the outcome probabilities of a Clifford circuit. 

The original proof of the Gottesman-Knill Theorem makes use of the \textit{stabilizer} formulation of quantum mechanics, in which the state of the system at each time step is represented not by the amplitudes of the state vector, but by a set of Pauli operators which stabilize it \cite{GottesmanThesis}. Using this approach, the problem of computing the outcome probabilities of Clifford circuits can be reduced to computing inner products between stabilizer states \cite{AaronsonGottesman}. The latter can be done efficiently using the stabilizer formalism, and hence the outcome probabilities can be computed efficiently.

Besides the stabilizer formalism, other techniques have been used to compute the outcome probabilities of Clifford circuits efficiently (for some examples, see \cite{Dehaene, VandenNest,  Jozsa, nest2012efficient}). In this paper, we present a different method from these that is explicitly based upon Feynman's sum-over-paths technique \cite{FeynmanHibbs, Dawson, Penney}. 
We restrict our attention to Clifford circuits acting on collections of quopits, i.e.,\ $p$-level systems where $p$ is an odd prime \cite{Emerson} (a few remarks about extending our results to qubit systems will be made in Section \ref{sec:conclusion}). In this approach, the amplitudes of quantum circuits are expressed in terms of a weighted sum over computational paths. 

For general quantum circuits, such a sum over paths involves an exponential number of terms, and no efficient algorithm exists to compute this sum, unless $\#\P$-complete problems can also be solved efficiently. However, building on the work of Dawson \textit{et al}.\@ \cite{Dawson}, we show that for quopit Clifford circuits, the sum over paths takes a special form: it can be expressed as a product of Weil sums \cite{Weil} with quadratic polynomials. The problem of evaluating Weil sums explicitly is in general difficult, but for Weil sums with quadratic polynomials, the sum can be computed efficiently. This gives an efficient algorithm to compute not just the outcome probabilities but also the amplitudes of quopit Clifford circuits when all $n$ quopit registers are measured. In other words, such circuits admit of an efficient $\strong(n)$ simulation~\cite{koh2015further} (see also Section \ref{sec:notions} for a discussion of various notions of simulation).

The sum-over-paths technique has previously been used to answer computational complexity questions about the power of quantum computation. For example, by considering quantum circuits comprising only gates from the universal gate set of Toffoli and Hadamard gates, Dawson \etal provide a simple proof of the complexity-theoretic result that $\BQP \subseteq \PP$ (first proved by \cite{Adleman}), one of the tightest `natural' upper bounds for $\BQP$ \cite{Dawson}. Dawson \etal then ask what other universal gate sets are amenable to the sum-over-paths approach. An extension of this question, that we address in this paper, is to ask not just about universal gate sets, but also about gate sets corresponding to restricted models of quantum computation. 

Another example is the class of linear algebraic quantum circuits (which are closely related to Clifford circuits) studied by Bacon \etal \cite{AlgCircuits}, who noted that the sum-over-paths technique introduced by Dawson \etal implied that the computation of outcome probabilities in such circuits (assuming all registers are measured) can be reduced to the computation of Weil sums for quadratic polynomials, implying efficient classical simulation of such circuits (specifically, efficient $\mathsf{STR}(n)$ simulation). When specialized to the case of quopits, however, the group of unitaries implementable in a linear algebraic quantum circuit is a proper subgroup of those implementable by a quopit Clifford circuit because the generating gate set does not include the phase gate ($R$ in \eq{def:cliffordgates}, which corresponds to a phase space squeezing operation).  In this respect, our result generalizes theirs.  Furthermore, we here provide an explicit expression for not just the outcome probabilities, as Bacon \etal do, but the amplitudes as well.

The sum-over-paths technique makes explicit a correspondence between quantum circuits and low-degree polynomials, known as the \textit{circuit-polynomial correspondence} \cite{Montanaro}. This correspondence can be exploited in two different directions. In the first direction, using quantum circuit concepts, it enables one to prove classical results about polynomials. For example, the Gottesman-Knill Theorem, which is a theorem about quantum circuits, can be used to provide an efficient algorithm to compute the gap of degree-2 polynomials over $\mathbb F_2$ \cite{Montanaro}. In the second direction, known classical results about polynomials can be used to provide algorithms for simulating classes of quantum circuits. Our result, in which we use classical results about degree-2 polynomials to simulate quopit Clifford circuits, provides an example of the second direction. Note that while the polynomials in \cite{Dawson} and \cite{Montanaro} are over $\mathbb F_2$, our results about quopit systems involve polynomials over the field $\mathbb F_p$ where $p$ is an odd prime.


The rest of the paper is structured as follows. In Section \ref{sec:prelim}, we introduce the relevant definitions and notations and describe the problem of interest. In Section \ref{sec:constructingPathSumExpressions}, we review the sum-over-paths technique and show how to construct sum-over-paths expressions for quopit Clifford circuits. In Section \ref{sec:evaluatingsums}, we show how the sum-over-paths expression can be computed classically in polynomial time. In Section \ref{sec:notions}, we discuss different notions of classical simulation and their relation to the Gottesman-Knill Theorem. In Section \ref{sec:balancedness}, we show how our results can be used to show that unitary operations implemented by quopit Clifford circuits are necessarily balanced. Finally, in Section \ref{sec:conclusion}, we conclude by discussing other gate sets, including qubit Clifford gates.

\section{Preliminary definitions and notation}
\label{sec:prelim}

In this paper, $p$ will always denote an odd prime. We shall work over the finite field $\mathbb{F}_p$ of characteristic $p$, which is the set of integers modulo $p$. The set of $n\times n$ matrices over $\mathbb F_p$ is denoted by $M_n(\mathbb F_p)$, and the group of invertible $n\times n$ matrices over $\mathbb F_p$ is denoted by $\mathrm{GL}_n(\mathbb F_p)$. 

We confine our attention to quopit systems, i.e.,\ $p$-level quantum systems where $p$ is an odd prime. A quopit Clifford circuit acting on quopit systems is defined to be any circuit consisting of only the following gates, called \textit{quopit Clifford gates}: the Fourier gate $F$, the phase gate $R$ and the sum gate $\Sigma$, which are the $p$-level generalizations, respectively, of the Hadamard, phase and CNOT gates of qubit Clifford circuits defined in \eq{qubitcircuit}. They are defined as follows:
\begin{eqnarray}
\label{def:cliffordgates}
F &\equiv& \frac 1{\sqrt{p}} \sum_{s,t \in \mathbb F_p} \chi(st) \ket s \bra t,  \nonumber \\
R &\equiv& \sum_{t\in\mathbb F_p} \chi(t(t-1)2^{-1}) \ket t \bra t, \nonumber \\
\Sigma &\equiv& \sum_{s,t\in \mathbb F_p} \ket{s,s+t}\bra{s,t},
\end{eqnarray}
where $\chi(a) \equiv \exp(2\pi i a/p)$, and $2^{-1} = (p+1)/2$ is the inverse of 2 modulo $p$. For the sum gate, we write $\Sigma_{ab}$ to indicate that $a$ and $b$ are the control and target registers respectively, i.e. $\Sigma_{ab} = \sum_{s,t\in \mathbb F_p} \ket{s}\bra{s}_a \otimes \ket{s+t}\bra{t}_b$.

For a given circuit, let $n$ denote the number of registers (i.e.\ number of quopits), and $N$ denote the number of gates. We make the following additional assumptions about the circuit (for an example, see the circuit diagram in Figure \ref{fig:labelingscheme1}): 
\begin{itemize}
\item The inputs to the circuit are computational basis states $\ket a$, where $a\in \mathbb F_p^n$. 
\item Measurements are performed only at the end of the circuit, i.e., there are no intermediate measurements, and \textit{all} quopits are measured at the end of the circuit. Also, measurements are performed in the computational basis. Hence, the possible measurement outcomes lie in the set $\mathbb F_p^n$. A measurement outcome of $b \in \mathbb F_p^n$ is associated with the computational basis vector $\ket b$.
\item There are no extraneous quopits, i.e. every quopit is acted on by at least one gate, so that $n=O(N)$.
\end{itemize}

The problem we are interested in, which we call $\mathcal P$, is the following: given a quopit Clifford circuit acting on the input state $\ket a$, where $a \in \mathbb F_p^n$, compute the probability amplitude associated with the outcome $b \in \mathbb F_p^n$. Formally, $\mathcal P$ may be stated as follows: \newline\newline
\textit{
Given a description of a quopit Clifford circuit that implements the unitary $U$, as well as strings $a, b \in \mathbb F_p^n$, compute $\langle b|U|a\rangle$.} 
\newline\newline
Note that a description of a quopit Clifford circuit $C$ is a specification of the gates in $C$ as well as the registers on which they act. 

If $C$ were allowed to be a general quantum circuit with gates chosen from some universal discrete gate set, then the problem $\mathcal P$ would be $\#\P$-hard. But for the quopit Clifford circuits $C$ that we consider, $\mathcal P$ can be solved in polynomial-time. We now describe a proof of this result that is based on the sum-over-paths formulation of quopit Clifford circuits.

\section{Constructing sum-over-paths expressions for quopit Clifford circuits}
\label{sec:constructingPathSumExpressions}

In this section, we review the sum-over-paths technique applied to quopit Clifford circuits that was introduced in Section III of Ref.~\cite{Penney}. Without loss of generality, we assume that each register of the Clifford circuit terminates in a Fourier gate just before it is measured (we shall refer to circuits with this property as \textit{standard-form quopit Clifford circuits}). If this were not the case, for each register that does not terminate in a Fourier gate, we could pad the circuit by inserting 4 Fourier gates before the measurement is performed, since $F^4 = \mathbb I$.  The Fourier gates that appear just before a measurement shall be called \textit{terminal} Fourier gates. All other Fourier gates will be called \textit{non-terminal}.

For a quopit Clifford circuit $C$ with input labeled by $a=a_1 \ldots a_n \in \mathbb F_p^n$ and measurement outcome labeled by $b=b_1 \ldots b_n \in \mathbb F_p^n$, we shall label wires of $C$ at every time step to create a \textit{labeled circuit} as follows (See Figures \ref{fig:labelingscheme1} and \ref{fig:labelingscheme2} for an example):
\def\lab#1{\ustick{{}_{#1}}}
\def\llab#1{ \qw & \ustick{{}_{#1}} \qw &\qw}
\def\lllab#1{\qw & \qw & \ustick{{}_{#1}} \qw &\qw & \qw}

\begin{figure}
\begin{align}
\Qcircuit @C=1em @R=1em {
& &\lstick{\ket{a_1}} & \gate{R}  & \multigate{1}{ \Sigma_{12}  } & \gate {F} & \gate{F} & \meter  & \!\!\! {}^{b_1} \\ 
& &\lstick{\ket{a_2}} & \gate{F}  & \ghost{ \Sigma_{12}  } & \multigate{1}{ \Sigma_{23}  } & \gate{F} & \meter & \!\!\! {}^{b_2} \\ 
& &\lstick{\ket{a_3}} & \qw  & \gate{F} & \ghost{ \Sigma_{23}  } & \gate{F} & \meter & \!\!\! {}^{b_3} \\ 
} \nonumber
\end{align}
\caption{\small Example of a quopit Clifford circuit. As explained in the text, we can assume without loss of generality that each register ends in a Fourier gate.}
\label{fig:labelingscheme1}
\begin{align}
\Qcircuit @C=1em @R=1em {
& &\lab{a_1} \qw& \gate{R} &\lab{a_1} \qw & \multigate{1}{ \Sigma_{12}  } & \llab{a_1}  & \gate {F} & \llab{x_3} & \gate{F} &\lab{b_1} \qw   \\ 
& &\lab{a_2} \qw & \gate{F}  &\lab{x_1} \qw & \ghost{ \Sigma_{12}  } & \llab{a_1 + x_1}  & \multigate{1}{ \Sigma_{23}  } &\llab{a_1 + x_1}  & \gate{F} &\lab{b_2} \qw  \\ 
& &\lab{a_3} \qw & \qw &\lab{a_3} \qw & \gate{F} & \llab{x_2} & \ghost{ \Sigma_{23}  } & \llab{a_1 + x_1 + x_2} & \gate{F} & \lab{b_3} \qw\\ 
} \nonumber
\end{align}
\caption{\small Labeled circuit corresponding to the circuit in Figure \ref{fig:labelingscheme1}. The phase polynomial is read off to be $S(x_1,x_2,x_3) = a_2 x_1 + a_3 x_2 + a_1 x_3 + x_3 b_1 + b_2(a_1+x_1)+b_3(a_1 + x_1 + x_2) + 2^{-1} a_1 (a_1-1).$}
\label{fig:labelingscheme2}
\end{figure}

\begin{enumerate}
\item Label the input wires by $a_1,\ldots, a_n$.
\item Going from left to right of the circuit diagram for $C$, label the wires at each subsequent time step as follows: 
\begin{enumerate}
\item \label{Rgate} For each phase gate $R$ and identity gate $\mathbb I$ (i.e.\ when we have a bare wire), if the label at the input is $s$, then we label the output by $s$. 
\item \label{sumGate} For each sum gate $\Sigma$, if the labels at the inputs are $(s,t)$, then label the outputs by $(s, s+t)$. Here the first element in the pair is the control register, and the second element in the pair is the target register.
\item For the $l$th non-terminal Fourier gate $F$, we introduce an auxiliary variable $x_l$, and regardless of the input to the Fourier gate, we label the output by $x_l$.
\end{enumerate}
\item Label the output wires by $b_1, \ldots, b_n$.
\end{enumerate}

\def\IN#1{\textrm{in}(#1)}
\def\OUT#1{\textrm{out}(#1)}

We shall associate each quopit Clifford circuit with a polynomial over $\mathbb F_p$, which is called the \textit{phase polynomial} \cite{gerdt2006algorithm}. The variables in the phase polynomial are the auxiliary variables $x = (x_1, \ldots, x_\alpha)$, where $\alpha$ is the number of non-terminal Fourier gates in the circuit.  For each gate $G$ in the circuit, let $\IN{G}$ and $\OUT{G}$ be the input and output labels of that gate in the labeled circuit. The phase polynomial associated with a quopit Clifford circuit is the polynomial $S$ over $\mathbb F_p$ defined by (see Figure \ref{fig:labelingscheme2} for an example):
\begin{eqnarray}
\label{def:phase}
S(x) &=& \sum_{ \textrm{Fourier gates} \ F} \IN{F} \OUT{F} \nonumber\\ &&+ \sum_{ \textrm{phase gates} \ R} 2^{-1} \, \IN{R}(\IN{R}-1).  
\end{eqnarray}

The theorem relating the circuit amplitudes and the phase polynomial, which appears as Theorem 3 in \cite{Penney}, is the following: 
\begin{theorem}
\label{mainTheorem}
Let $C$ be a standard-form quopit Clifford circuit on $n$ registers that implements the Clifford operation $U$. Let $\alpha$ be the number of non-terminal Fourier gates and let $S(x)$ be the phase polynomial associated with $C$. Then,
\begin{equation} 
\label{mainFormula} 
\langle b|U|a \rangle  = \frac{1}{p^{(n+\alpha)/2}} \sum_{x \in \mathbb F_p^\alpha} \chi(S(x)).  
\end{equation} 
\end{theorem}

\section{Evaluating the sum over paths}
\label{sec:evaluatingsums}

Using Theorem \ref{mainTheorem}, the problem $\mathcal P$ is reduced to evaluating the sum in \eq{mainFormula}. In this section, we describe how this sum may be evaluated.

First, we note that $S(x)$ is a degree-2 polynomial in the variables $x= (x_1,\ldots, x_\alpha)$ as well as the variables $a_1,\ldots, a_n, b_1,\ldots, b_n$. This is due to the fact that $S(x)$ is a sum of terms which are at most quadratic, since the input and output labels of each gate are linear in the variables $x_i, a_i$ and $b_i$. Hence, we can write $S(x)$ as 
\begin{equation}
\label{expforAction}
S(x) = x^T \Theta x + \eta^T x + \zeta =  \sum_{i,j=1}^\alpha  \Theta_{ij} x_i x_j + \sum_{i=1}^\alpha \eta_i x_i + \zeta,
\end{equation}
where $\Theta \in M_\alpha(\mathbb F_p)$ can be chosen to be symmetric, $\eta \in \mathbb F_p^\alpha$ and $\zeta \in \mathbb F_p$. Note that while $\eta$ and $\zeta$ are dependent on $a$ and $b$, $\Theta$ is independent of $(a,b)$, because otherwise $S(x)$ as a polynomial in 
 $x_i, a_i$ and $b_i$ would have a degree that exceeds 2.

Substituting \eq{expforAction} into \eq{mainFormula} gives
\begin{equation} 
\label{mainFormula1} 
\langle b|U|a \rangle  = \frac{\chi(\zeta)}{p^{(n+\alpha)/2}} \sum_{x \in \mathbb F_p^\alpha} \chi(x^T \Theta x + \eta^T x ) .
\end{equation}
The above sum can be evaluated using the following two steps. 
\subsection{Step 1: Diagonalizing \texorpdfstring{$\Theta$}{Theta}}
In the first step, we diagonalize the matrix $\Theta$, by making use of the following theorem:
\begin{theorem}
\label{thm:diagonalization}
There is a polynomial-time algorithm $T$ that when given a symmetric matrix $\Theta \in M_{\alpha}(\mathbb F_p)$ outputs an invertible matrix $L \in \mathrm{GL}_\alpha(\mathbb F_p)$ such that $L^T \Theta L$ is diagonal.
\end{theorem}
\begin{proof}
The proof is essentialy an algorithmic implementation of the standard proof that any quadratic form over a field that is not of characteristic 2 is diagonalizable (see Propositions 6.20 and 6.21 of \cite{Lidl}). We present a polynomial-time algorithm in Appendix \ref{sec:proofOfDiagonalization}.
\end{proof}

By making use of Theorem \ref{thm:diagonalization}, and the change of variables $\mu = L^T \eta$ and $x = Ly$, we can rewrite $x^T \Theta x + \eta^T x = y^T \Lambda y + \mu^T y$, where $\Lambda = L^T \Theta L$ is a diagonal matrix. By this change of variables, the sum in \eq{mainFormula1} becomes
\begin{align}
\label{singlevariablesum}
\sum_{x \in \mathbb F_p^\alpha} \chi(x^T \Theta x + \eta^T x) 
&=  \sum_{y_1,\ldots, y_\alpha  \in \mathbb F_p} \chi\left( \sum_{i=1}^\alpha \lambda_i y_i^2 + \mu_i y_i\right) \nonumber\\ 
&=  \prod_{i=1}^\alpha \sum_{y_i \in \mathbb F_p} \chi\left(\lambda_i y_i^2 + \mu_i y_i\right) ,
\end{align}
where the $\lambda_i$ are the diagonal entries of $\Lambda$, and the $\mu_i$ are the components of $\mu$. We see from \eq{singlevariablesum} that one needs only to compute the (much simpler) sums over a single variable. Such sums are called \textit{Weil sums} \cite{Lidl}, and we will show in the next step how to compute them.

\subsection{Step 2: Using the exponential sum formula}

The second step makes use of the following theorem about exponential sums (see Theorem 5.33 of \cite{Lidl}):
\begin{theorem}
\label{thm:expsum}
The sum $$\sum_{y \in \mathbb F_p} \chi\left(\lambda y^2 + \mu y\right)$$ can be explicitly evaluated as follows: 
\begin{enumerate}
\item If $\lambda = \mu = 0$, then it equals $p$.
\item If $\lambda =0$ and $\mu \neq 0$, then it equals $0$.
\item If $\lambda \neq 0$, then it equals
\begin{equation}
i^{\varepsilon(p)} \chi\left(-4^{-1} \lambda^{-1}\mu^2\right) \left(\frac{\lambda}{p} \right)  \sqrt{p},
\end{equation}
where $\left(\frac{\lambda}{p} \right)$ is the Legendre symbol and $\varepsilon(p) =0$ if $p$ is congruent to $1$ mod $4$ and $\varepsilon(p) =1$ otherwise. The inverses are modular multiplicative inverses modulo $p$.
\end{enumerate}
\end{theorem}

If we partition the indices $i \in \{1,\ldots,\alpha\}$ into the following sets
\begin{eqnarray}
X &=& \{i \in \{1,\ldots,\alpha\}| \lambda_i \neq 0\},  \nonumber\\
Y &=& \{i \in \{1,\ldots,\alpha\}| \lambda_i = 0, \mu_i = 0\}, \nonumber\\ 
Z &=& \{i \in \{1,\ldots,\alpha\}| \lambda_i = 0, \mu_i \neq 0\},  \nonumber
\end{eqnarray}
then by using the exponential sum formula in Theorem \ref{thm:expsum} to evaluate the sums in \eq{singlevariablesum}, we obtain 
\begin{eqnarray} 
\label{finalForm}
\langle b|U|a \rangle  &=& 
\frac{\chi(\zeta)}{p^{(n+\alpha)/2}} \left( \prod_{i\in X} i^{\varepsilon(p)} \chi(-4^{-1} \lambda_i^{-1} \mu_i^2) \left(\frac{\lambda_i}{p}\right) \sqrt{p}\right) \nonumber\\
&&\times \left( \prod_{j\in Y} p\right) \left( \prod_{k\in Z} 0\right) \nonumber\\
&=& p^{-(n+r -\alpha)/2} \delta_{0,|Z|} i^{r \varepsilon(p)} \left( \frac{ \prod_{i\in X} \lambda_i}{p} \right) \nonumber\\ &&\times \  \chi\left(\zeta - 4^{-1} \sum_{i\in X} \lambda_i^{-1} \mu_i^2\right),
\end{eqnarray}
where $r= |X|$ is the rank of $\Theta$, i.e.\ the number of nonzero diagonal entries in $\Lambda$. Note that we used the multiplicative property of the Legendre symbol: $\left( \frac{ \prod_{i\in X} \lambda_i}{p} \right) = \prod_{i\in X} \left( \frac{\lambda_i}{p} \right)$, and the fact that when $|Z|=0$, $|X|+|Y|=\alpha$. Here, $\delta_{x,y}$ is the Kronecker delta.

Now, the Legendre symbol $\left(\frac{a}{p}\right)$ takes values in the set $\{-1,0,1\}$ and vanishes only when $a \equiv 0 \mod p$. Since $\lambda_i \neq 0$ for $i\in X$, by definition, we get the following simple expression for the outcome probabilities:
\begin{equation} \label{expressionForProbabilities}
|\langle b|U|a \rangle|^2 = \frac{1}{p^{n+r-\alpha}}  \delta_{0,|Z|}.
\end{equation}

\subsection{Running time}

We now analyze the running time of the above procedure. The evaluation of the matrix element $\bra b U \ket a$ involved four main steps: First, given a decomposition of $U$ in terms of quopit Clifford gates, we employed the labeling procedure described in Section \ref{sec:constructingPathSumExpressions} to label the Clifford circuit. Second, from the labeled circuit, we computed the phase polynomial $S(x)$ defined by \eq{def:phase}. Third, we used Theorem \ref{thm:diagonalization} to diagonalize $S(x)$, and fourth, we used the exponential sum formula in Theorem \ref{thm:expsum} to calculate the matrix element in \eq{finalForm}. 

Steps 1, 2, and 4 take time that is linear in the size of the circuit. Step 3 involves matrix diagonalization 
which can be carried out in polynomial time, as we show
in Appendix \ref{sec:proofOfDiagonalization}. Hence, the algorithm that we give here to compute $\bra b U \ket a$ runs in polynomial time. 

\section{Notions of classical simulation and the Gottesman-Knill Theorem}
\label{sec:notions}

The Gottesman-Knill Theorem is a result stating that Clifford circuits can be efficiently simulated by a classical computer. Since its first appearance in \cite{gottesman1998heisenberg}, several other variants and generalizations of the theorem have been found, and today the term is often used, loosely, to refer to any one of a collection of results that are variants of the original theorem \cite{nielsen2010quantum,aaronson2004improved,van2010classical,Jozsa,koh2015further}. These variants vary according to the ingredients of the Clifford circuit and what it means to simulate it. The goal of this section is to clarify some of the differences between these variants and to discuss the relationship between the Gottesman-Knill Theorem and our results from Section \ref{sec:evaluatingsums} that were obtained via the sum-over-paths approach.

Two common notions of classical simulation of quantum computation are \textit{weak} and \textit{strong} simulation \cite{terhal2004adptive,VandenNest}. Let $T$ be a description of a quopit quantum circuit with $n$ registers, which we index by the integers in $[n] = \{1,\ldots,n\}$. Let $I = \{i_1,\ldots, i_{|I|}\} \subseteq [n]$ be a subset of indices, and let $y_{|I|}\in \mathbb F_p^{|I|}$ be a $|I|$-tuple of elements from $\mathbb F_p$. Define $p_T^{I} (y_{|I|})$ to be the probability that the outcomes $y_{|I|}$ are observed when the registers whose indices are in the set $I$ of the quantum circuit $T$ are measured. A strong simulation of a family of quopit circuits is a deterministic classical algorithm that takes as input a triple $\langle T, I, y_{|I|}\rangle$ and outputs the probability $p_T^{I} (y_{|I|})$. A weak simulation of a family of quopit circuits is a randomized classical algorithm that takes as input the pair $\langle T, I\rangle$ and outputs the values $y \in \mathbb F_p^{|I|}$ according to the probability distribution $p_T^{I}$. Stated informally, a strong simulation involves computing the marginal probabilities of the measurement outcomes of a quantum circuit, while a weak simulation involves just sampling from the same distribution as the quantum circuit. It was shown in \cite{terhal2004adptive} that strong simulation implies weak simulation.

Note that for the notions of strong and weak simulations, no restrictions are placed on the size of the subsets $I$ that are fed as inputs to the algorithm. To describe finer-grained notions of simulation that take into account the size of the subset of outputs to be simulated, we adopt the terminology introduced by Koh \cite{koh2015further}: Let $f(n)$ be a function of the number of registers $n$ in the circuit. Similar to strong simulation, a $\strong(f(n))$-simulation (which stands for strong-$f(n)$ simulation) of a family of quopit circuits is a deterministic classical algorithm that takes as input a triple $\langle T, I, y\rangle$ and outputs the probability $p_T^{I} (y)$, with the restriction that the subset $I = \{i_1,\ldots, i_{f(n)}\} \subseteq [n]$ and the tuple $y\in \mathbb F_p^{f(n)}$ must be of size $f(n)$. Similarly, a $\weak(f(n))$-simulation is defined the same way as weak simulation, except that the subset $I$ is required to be of size $f(n)$. The class that is relevant to our results is $\strong(n)$, which is obtained by taking $f(n) = n$. While an efficient strong simulation implies an efficient $\strong(n)$-simulation, the latter seems incomparable to efficient weak simulation \cite{koh2015further}. 

The original Gottesman-Knill Theorem states that there exists an efficient weak simulation of adaptive qubit Clifford circuits with computational basis inputs and computational basis measurements. Here, an efficient simulation is one that can be carried out in polynomial-time, and an adaptive circuit is one which has intermediate measurements whose outcomes may affect which operations we perform next.  By changing the ingredients of the Clifford circuit as well as considering different functions $f(n)$ in the definitions of $\strong(f(n))$ and $\weak(f(n))$ simulations, several variants of the Gottesman-Knill Theorem can be obtained \cite{koh2015further, Jozsa}.  Not all variations of ingredients and notions of simulation lead to efficient classical simulability though. In fact, the situation is much more delicate. As observed by Jozsa and Van den Nest \cite{Jozsa} and Koh \cite{koh2015further}, small changes to the ingredients of a Clifford circuit can lead to dramatic changes in the simulation complexity. For example, while nonadaptive qubit Clifford circuits with computational basis inputs and computational basis measurements are efficiently strongly simulable, adaptive qubit Clifford circuits with all the other ingredients kept the same are unlikely to be efficiently strongly simulable (an efficient strong simulation of such circuits can be shown to be $\# \P$-hard \cite{Jozsa}), even though it might seem that adaptability is a \textit{classical} resource. A table classifying which combinations of ingredients and which notions of simulations lead to efficient simulations is presented in \cite{koh2015further} (which generalizes a similar table found in \cite{Jozsa}).

Using the terminology described above, our main result may be stated as follows.
\begin{theorem}
\label{thm:mainThmx}
There exists an efficient $\strong(n)$-simulation of nonadaptive quopit Clifford circuits with computational basis inputs and computational basis measurements. In particular, the amplitude associated with starting with a computational basis state $\ket a$ as the input to the circuit, where $a \in \mathbb F_p$, and measuring the result $b \in \mathbb F_p$, is given by \eq{finalForm}, and the corresponding probability is given by \eq{expressionForProbabilities}.
\end{theorem}

A few remarks are in order. First, we note that Theorem \ref{thm:mainThmx} corresponds to Case (iv) of Table 1 in \cite {koh2015further}, if the results in the paper are extended to quopit Clifford circuits. Secondly, it is unlikely that we can extend the theorem to include adaptive quopit Clifford circuits, since an analogous theorem to Theorem 2 of \cite{koh2015further}, which states that  $\strong(n)$-simulation of adaptive Clifford circuits is $\# \P$-hard, should hold true for quopit Clifford circuits. Thirdly, we leave open the question about whether the sum-over-paths approach is useful for proving efficient $\strong(1)$ or strong simulation of quopit Clifford circuits (we know that these exist though, by the Gottesman-Knill Theorem). Note that Theorem \ref{thm:mainThmx} does not immediately imply either efficient $\strong(1)$ or efficient strong simulation since, in general, computing a marginal probability from a joint probability requires summing an exponential number of terms and cannot be performed efficiently, unless there is some structure in the problem.

\section{Balancedness of quopit Clifford circuits}
\label{sec:balancedness}

Quopit Clifford gates have the property that their nonzero matrix elements relative to the computation basis have the same absolute value. A gate with this property (and the matrix representing it) is called \textit{balanced} \cite{Dawson, Penney}.

\begin{definition} \label{def:balanced}
 A gate $G$ acting on $n$ qudits represented by a unitary $U_G$ is balanced if there is a constant $c \in \mathbb R_{\geq0}$ and functions $f:\left(\zx_d\right)^n \times \left(\zx_d\right)^n \to \mathbb R$ and $g:\left(\zx_d\right)^n \times \left(\zx_d\right)^n \to \left(\zx_d\right)^n$ such that for all $a,b \in \left(\zx_d\right)^n$,
 \begin{equation}
  \bra b U_G \ket a = c \, e^{i f(a,b)} \delta_{0,g(a,b)}.
  \end{equation}
\end{definition}


As noted in \cite{Penney}, only if all gates in a circuit are balanced can one use the sum-over-paths technique to evaluate its functionality. We shall refer to the number $c$ as the \textit{weight} of the gate $G$. By convention, whenever the basis is not specified, it is assumed that the balanced property is defined with respect to the computational basis. Hence, the Fourier, phase and sum gates defined in \eq{def:cliffordgates} are balanced with weights $p^{-1/2}$, $1$ and $1$ respectively.

Unitary operations implemented by circuits consisting of balanced gates are not balanced in general. For example, consider a circuit consisting of the Hadamard gate $H$, given by \eqref{qubitcircuit}, and the gate $V$, given by
$$ V = \frac 1{\sqrt 2} \left(\begin{matrix}
e^{-i \theta} & -e^{i \theta} \\ e^{-i \theta} & e^{i \theta}
\end{matrix} \right). $$
It is straightforward to check that both $H$ and $V$ are balanced and unitary. The product $VH$, however, is
$$ VH =  \left(\begin{matrix}
-i \sin(\theta) & \cos(\theta) \\ \cos(\theta) & -i \sin(\theta)
\end{matrix} \right), $$
which is not balanced in general. For example, when $\theta = \pi/6$, the entries of $VH$ have absolute values $1/2$ and $\sqrt{3}/2$.

For quopit Clifford circuits, however, unitary operations implemented by quopit Clifford gates are always balanced. This is a direct consequence of \eq{expressionForProbabilities}. To see this, recall that $\Theta$ (defined in Eq.~\eqref{expforAction}) is independent of $(a,b)$. Hence, $r = \mathrm{rank} (\Theta)$ is independent of $(a,b)$, which implies that the nonzero terms of $|\langle b|U|a \rangle|$, which are equal to $p^{-(n+r-\alpha)/2}$, are independent of $(a,b)$. This result is summarized in the following theorem:

\begin{theorem}
Let $U$ be a unitary operation on $n$ quopits that is implemented by a standard-form Clifford circuit $C$ with $\alpha$ non-terminal Fourier gates, and phase polynomial $S(x)$. Let $r$ be the rank of the coefficient matrix of the quadratic form corresponding to the degree-2 terms in $S(x)$. Then $U$ is a balanced matrix with weight $p^{- (n+r-\alpha)/2}$.
\end{theorem}

\section{Concluding remarks}
\label{sec:conclusion}

The sum-over-paths approach has proved useful in providing an efficient algorithm for computing the amplitudes of circuits composed of quopit Clifford gates. The approach, however, is by no means applicable to only quopit Clifford circuits -- our proof of Theorem \ref{thm:mainThmx} can easily be extended to imply a stronger result. To see this, note that the arguments presented in the proof depended on only the following two properties of the set $\mathcal S$ of quopit Clifford gates:

\begin{enumerate}
 \item Each gate $G\in \mathcal S$ is balanced with the function $g(a,b)$ being linear and the function $f(a,b)$ having the form 
 \begin{equation}
 f(a,b) = \frac{2 \pi}{p} S_G(a,b),
 \end{equation}
 where $S_G$ is a polynomial of degree at most 2 in the components of $a$ and $b$ having coefficients in $\mathbb F_p$. 
\item There is a gate $G \in \mathcal S$ with finite order (i.e. there exists a natural number $k$ such that $(U_G)^k = \mathbb I$) that satisfies $g(a,b) = 0$ for all $a,b$.
 \end{enumerate}

For any class of circuits composed of gates satisfying these two properties, the transition amplitudes can be computed by a sum-over-paths expression of the form 
\begin{equation}
 \mathcal C \sum_{x \in \mathbb F_p^\beta} \chi(S(x)),
\end{equation}
where $\mathcal C \in \mathbb R_{>0}$ and $\beta \in \mathbb N$ are constants determined by the particular circuit, and $S(x)$ is a degree-$2$ polynomial in $\beta$ variables having coefficients in $\mathbb F_p$. 
Theorem \ref{thm:expsum} then gives us an efficient $\strong(n)$-simulation for such a class of circuits. In particular this includes the linear algebraic quantum circuits of Bacon \etal \cite{AlgCircuits}.

%
%

We conclude with a discussion about circuits with gate sets that do not satisfy the above two properties. Two classes of such circuits are of interest here. The first class of circuits are those composed of gates from universal gate sets. For example, consider a circuit composed of Toffoli and Hadamard gates. In general, the phase polynomial corresponding to such circuits is of degree 3, and hence does not satisfy Property 1 -- the sum-over-paths technique that leads to efficient simulation presented in this paper does not apply to such circuits. Based on conjectures in computational complexity theory, this result is expected since it is believed that an efficient classical simulation of (universal) quantum circuits is not possible. The second class of circuits are efficiently simulable circuits like the \textit{qubit Clifford circuits}, which are circuits composed of the Hadamard gate $H$, the phase gate $P$, and CNOT gate, defined as follows:
\begin{eqnarray}
\label{qubitcircuit}
H &=& \frac 1{\sqrt{2}} \sum_{s,t \in \mathbb F_2} (-1)^{st} \ket s \bra t, \nonumber \\
P &=& \sum_{t\in\mathbb F_2} i^t \ket t \bra t, \nonumber \\
\textrm{CNOT} &=& \sum_{s,t\in \mathbb F_2} \ket{s,s+t}\bra{s,t}.
\end{eqnarray}
From the above expression for the phase gate $P$, it is evident that we cannot find a polynomial $S_P(s,t)$ over $\mathbb F_2$ such that the  matrix elements $\langle s |P |t\rangle$ that are nonzero are of the form $\chi(S_P(s,t))$. As discussed in \cite{Penney}, one needs to define the phase polynomial to be over $\mathbb Z_4$ instead. We note also that for qubit Clifford circuits, there are other parts of the proof which would fail, for example, Theorem \ref{thm:diagonalization}, which works only for fields that are not of characteristic 2. For a treatment of evaluating amplitudes of qubit Clifford circuits using exponential sums, we refer the reader to \cite{bravyi2016improved}.

\section*{Acknowledgements} 
DEK thanks Siong Thye Goh for helpful discussions. DEK is supported by the National Science Scholarship from the Agency for Science, Technology and Research (A*STAR). MDP is supported by NSERC through the Doctoral Postgraduate Scholarship and by Corpus Christi College, Oxford. Part of this research was conducted during the 2015 Convergence Conference at the Perimeter Institute for Theoretical Physics.  Research at the Perimeter Institute for Theoretical Physics is supported in part by the Government of Canada through NSERC and by the Province of Ontario through MRI.

\bibliographystyle{unsrt}
\bibliography{Bib}
\newpage
\onecolumngrid
\appendix
\section{Proof of Theorem \ref{thm:diagonalization}}
\label{sec:proofOfDiagonalization}

\def\aligning#1#2{\parbox[t]{\dimexpr\linewidth-#1\algorithmicindent/16}{ \raggedright #2 \strut}}
\algnewcommand\algorithmicinput{\textbf{Input:}}
\algnewcommand\INPUT{\item[\algorithmicinput]}
\algnewcommand\algorithmicoutput{\textbf{Output:}}
\algnewcommand\OUTPUT{\item[\algorithmicoutput]}
\algnewcommand\algorithmicxprocedure{\textbf{Procedure:}}
\algnewcommand\PROCEDURE{\item[\algorithmicxprocedure]}

We shall describe a polynomial-time algorithm $T$ that, when given a symmetric matrix $\Theta \in M_\alpha(\mathbb F_p)$, outputs an invertible matrix $L \in \mathrm{GL}_\alpha(\mathbb F_p)$ such that $L^T \Theta L$ is diagonal. As we assumed in the main text, $p$ denotes an odd prime. The proof is essentialy an algorithmic implementation of Propositions 6.20 and 6.21 of \cite{Lidl}. 

We use the following notation in the proof: The matrix direct sum is denoted by $A \oplus B = \mathrm{diag}(A,B)$. The $n \times n$ identity matrix is denoted by $\mathbb I_n$. The Kronecker delta is denoted by $\delta_{ij}$. The components of a matrix $A\in  M_n(\mathbb F_p)$ are denoted by $A_{ij}$, with the indices taking values $i,j=1,\ldots,n$. Likewise, the components of a vector $v \in \mathbb F_p^n$ are denoted by $v_i$, with $i=1,\ldots,n$. The inverse $a^{-1}$ is defined to be the multiplicative inverse modulo $p$ of $a\in\mathbb F_p$, i.e.\ the unique $b\in \mathbb F_p$ for which $ab \equiv 1 \, (\mathrm{mod} \, p)$.

We first describe a subroutine, termed Algorithm 1, that we will need to call repeatedly. 

\begin{algorithm}
\caption{Subroutine for Algorithm $T$}
\label{alg:T1}
\begin{algorithmic}[1]
\INPUT a symmetric matrix $A \in M_n(\mathbb F_p)$, where $n\in \mathbb Z^+$.
\OUTPUT a 3-tuple $(P,a,B)$, where $P \in \mathrm{GL}_n(\mathbb F_p)$, $a\in \mathbb F_p$ and $B\in M_{n-1}(\mathbb F_p)$ is a symmetric matrix, such that
\begin{equation}
\label{outputreq}
P^T A P = a \oplus B.
\end{equation}
\If {$A=0$}
\State output $(\mathbb I_n,0,0)$.
\Else { \aligning{5}{let $a \in \mathbb F_p \backslash \{0\}$ and $c \in \mathbb F_p^n \backslash \{0\}$ be defined as follows:}
}
\State Check if there exists $i$ such that $A_{ii}\neq 0$.
\If {YES}
    \State Let $I = \min\{i|A_{ii}\neq 0\}$. 
    \State Set $a=A_{II}$, $c_j = \delta_{Ij}$ for $j=1,\ldots,n$.
\Else { NO}
    \State \aligning{8}{Let $(I,J) = \min\{(i,j)|A_{ij}\neq 0\}$. (where the minimum is taken with respect to some lexicographic ordering) }
    \State Set $a=2A_{IJ}$, $c_k = \delta_{Ik} + \delta_{Jk}$ for $k=1,\ldots,n$.
\EndIf
\State Choose $M$ for which $c_M \neq 0$.
\State Construct the nonsingular matrix $C \in \mathrm{GL}_n(\mathbb F_p)$ defined by: 
\begin{equation} \label{definingC}
C_{ij} = \begin{cases} c_i & j=1 \\ \delta_{i+1,j} & j \neq 1, i<M \\ 0 & j \neq 1, i=M  \\ \delta_{i,j} & j \neq 1, i>M. \end{cases}
\end{equation}
\quad \, for $i,j=1,\ldots,n$.
\State Compute 
\begin{equation} \label{definingb}
b_l = \sum_{ij} c_i A_{ij} C_{jl}, 
\end{equation}
\quad \, for $l=1,\ldots,n$.
\State Define 
\begin{equation} \label{definingg}
g(y_2, \ldots, y_n) = \sum_{k>1, l>1} \left(\sum_{ij} A_{ij} C_{ik} C_{jl}\right) y_k y_l - a^{-1} \left(\sum_{i>1} b_i y_i\right)^2.
\end{equation}
\State Construct the matrix $D \in \mathrm{GL}_n(\mathbb F_p)$ defined by 
\begin{equation} \label{definingD}
D_{ij} = \begin{cases} -b_j a^{-1} & i=1, j\neq 1 \\ \delta_{i,j} & \textrm{otherwise.} \end{cases}
\end{equation}
\State Compute the coefficient matrix $B \in M_{n-1}(\mathbb F_p)$ of the quadratic form $g$ using 
\begin{equation}
B_{ij} = 2^{-1}\left[g(e_i + e_j) - g(e_i) - g(e_j)\right],
\end{equation}
\quad \, where $e_i$ is the $i$th unit vector.
\State Compute $P=CD \in \mathrm{GL}_n(\mathbb F_p)$.
\State Output $(P,a,B)$.
\EndIf
\end{algorithmic}
\end{algorithm}

\noindent \textbf{Proof of correctness:} We shall prove that Algorithm \ref{alg:T1} works as described. If $A=0$, then $(P,a,B)=(I,0,0)$, which satisfies \eq{outputreq}, as required. Hence, for the rest of the proof, we shall assume that $A \neq 0$.

We first claim that $c^T A c = a$. Indeed, in the YES case in Line 5, $c^T A c = \sum_{kl} A_{kl} \delta_{ik} \delta_{il} = A_{ii} = a$, and in the NO case in Line 8, we have $A_{ii}=0$ for all $i$. Hence, $c^T A c = \sum_{kl} A_{kl} (\delta_{ik} + \delta_{jk})(\delta_{il} + \delta_{jl}) = A_{ij} + A_{ji} = 2A_{ij} = a$. Note that in both cases, $a\neq 0$. Hence $a^{-1}$ exists.

Next, consider the quadratic form $f$ corresponding to the matrix $A$:
\begin{eqnarray}
f(t_1, \ldots, t_n) = \sum_{ij} A_{ij} t_i t_j = t^T A t.
\end{eqnarray} 

Define 
\begin{eqnarray} \label{definingfprime}
f'(y_1,\ldots,y_n) = f(Cy) = \sum_{ij} A_{ij} (Cy)_i (Cy)_j.
\end{eqnarray} 

By expanding \eq{definingfprime}, and using \eq{definingC}, \eq{definingb} and \eq{definingg}, we obtain
\begin{eqnarray}
f'(y_1,\ldots,y_n) = a \left(y_1 + a^{-1} \sum_{l>1} b_l y_l \right)^2 + g(y_2,\ldots,y_n). 
\end{eqnarray}

Consider the matrix $D$ defined in \eq{definingD}. It is easy to see that its inverse has components given by
\begin{equation} \label{definingDinverse}
D_{ij}^{-1} = \begin{cases} b_j a^{-1} & i=1, j\neq 1 \\ \delta_{i,j} & \textrm{otherwise.} \end{cases}
\end{equation}

Let $x = D^{-1}y$. Then $x_1 = y_1 + a^{-1} \sum_{l>1} b_l y_l$ and  $x_i = y_i$, for all $i>1$. Hence,
\begin{equation} \label{expressionforfprime}
f'(y_1,\ldots,y_n) = a x_1^2 + g(x_2,\ldots, x_n).
\end{equation}

Now the LHS of \eq{expressionforfprime} is equal to
\begin{equation}
\label{LHSeq}
f'(y) = f'(Dx) = f(CDx) = f(Px) = x^T P^T A P x,
\end{equation}
and the RHS of  \eq{expressionforfprime} is equal to
\begin{eqnarray} \label{RHSeq}
a x_1^2 + g(x_2,\ldots, x_n) = ax_1^2 + \sum_{i,j=1}^n B_{ij} x_i x_j = x^T(a \oplus B) x,
\end{eqnarray}
where we used the fact that $B$ is the coefficient matrix of $g$.

Equating \eq{LHSeq} and \eq{RHSeq} then gives $ x^T P^T A P x = x^T (a \oplus B) x$. Since this holds for all $x$, we obtain
\begin{equation}
P^T A P = a \oplus B.
\end{equation}
\qed

We are now ready to describe the algorithm $T$:

\begin{algorithm}[h!]
\caption{Algorithm $T$ for matrix diagonalization}
\label{alg:T}
\begin{algorithmic}[1]
\INPUT a symmetric matrix $\Theta \in M_\alpha(\mathbb F_p)$, where $\alpha \geq 1$.
\OUTPUT an invertible matrix $L \in \mathrm{GL}_\alpha(\mathbb F_p)$ such that $L^T \Theta L$ is diagonal.
\If {$\alpha=1$}
\State Output $L=\mathbb I_1 \in M_1(\mathbb F_p)$.
\Else { Set $\Theta_\alpha = \Theta$.
}
\For {$k = \alpha, \alpha -1 ,\ldots, 2$}
\State Run Algorithm \ref{alg:T1} on $\Theta_k$ to get output $(P_k, a_k, \Theta_{k-1})$.
\EndFor
\For {$s = 1,\ldots,\alpha$}
\State Set $\tilde{P}_s = \mathbb I_{\alpha-s} \oplus P_s$.
\EndFor
\State Output $L = \tilde{P}_\alpha \tilde{P}_{\alpha-1} \ldots \tilde{P}_2$.
\EndIf
\end{algorithmic}
\end{algorithm}

\noindent \textbf{Proof of correctness:} If $\alpha=1$, then $\Theta$ is already diagonal. Hence, setting $L=\mathbb I$ to be the identity matrix gives the diagonal matrix $L^T \Theta L = \Theta$. Otherwise, we note that for the FOR loop with variable $k$ in Line 4 of Algorithm \ref{alg:T}, the output of Algorithm \ref{alg:T1} on $\Theta_k$ is the 3-tuple $(P_k, a_k, \Theta_{k-1})$ that satisfies
\begin{align} \label{actionOfT1}
P_k^T \Theta_k P_k = a_k \oplus \Theta_{k-1}.
\end{align}
Now it is straightforward to show by induction that
\begin{equation}
\tilde{P}_{\alpha-s}^T \ldots \tilde{P}_{\alpha-1}^T \tilde{P}_{\alpha}^T \Theta_\alpha \tilde{P}_{\alpha} \tilde{P}_{\alpha-1} \ldots \tilde{P}_{\alpha-s} = a_\alpha \oplus a_{\alpha-1} \oplus \ldots \oplus a_{\alpha-s} \oplus \Theta_{\alpha-s-1},
\end{equation}
for all $s = 0, 1, \ldots, \alpha-2$.

Hence, by using $s=\alpha-2$ and $\Theta_\alpha = \Theta$, we get
\begin{equation}
\tilde{P}_2^T \ldots \tilde{P}_{\alpha-1}^T \tilde{P}_{\alpha}^T \Theta \tilde{P}_{\alpha} \tilde{P}_{\alpha-1} \ldots \tilde{P}_2 = a_\alpha \oplus a_{\alpha-1} \oplus \ldots \oplus a_2 \oplus \Theta_1.
\end{equation}
Since each $\tilde{P}_s$ is invertible, their product $L = \tilde{P}_\alpha \tilde{P}_{\alpha-1} \ldots \tilde{P}_2$ is also invertible. Therefore, writing  $a_1 = \Theta_1 \in \mathbb F_p$, we get that
\begin{equation}
L^T \Theta L = \mathrm{diag}(a_\alpha, \ldots, a_2, a_1)
\end{equation}
is a diagonal matrix.

It is straightforward to see that Algorithm \ref{alg:T} runs in polynomial time in the size of the matrix $\alpha$.
\end{document}